\documentclass[letterpaper, 10 pt, conference]{ieeeconf}  
\IEEEoverridecommandlockouts                              
\usepackage{graphics} 
\usepackage{graphicx,color}
\usepackage[active]{srcltx}
\usepackage{amsmath} 
\usepackage{amssymb}  
\usepackage{amsmath,amssymb,amsfonts,theorem}
\usepackage{enumerate}
\usepackage{tikz}
\usepackage{tabu}
\usepackage{algpseudocode}

\usepackage{epstopdf}

{ \theorembodyfont{\normalfont} 

\newcommand{\R}{\mathbb{R}}

\newcommand{\eps}{\varepsilon}

\newcommand{\diag}{\operatorname{diag}} 

\newcommand\ie{\emph{i.e.}}

{
\newtheorem{example}{Example}
\newtheorem{remark}{Remark}
}
\newtheorem{assumption}{Assumption}
\newtheorem{definition}{Definition}

\newtheorem{lemma}{Lemma}

\newtheorem{proposition}{Proposition}

\def\QEDclosed{\mbox{\rule[0pt]{1.3ex}{1.3ex}}} 
\def\QEDopen{{\setlength{\fboxsep}{0pt}\setlength{\fboxrule}{0.2pt}\fbox{\rule[0pt]{0pt}{1.3ex}\rule[0pt]{1.3ex}{0pt}}}}
\def\QED{\QEDopen}

\def\be{\begin{equation}}
\def\ee{\end{equation}}
\def\ba{\begin{array}}
\def\ea{\end{array}}
\def\eqa{\begin{eqnarray}}
\def\eqe{\end{eqnarray}}
\title{\LARGE \bf Synchronization of Kuramoto oscillators in a\\ bidirectional frequency-dependent tree network}
\author{Matin Jafarian, Xinlei Yi, Mohammad Pirani, Henrik Sandberg, Karl Henrik Johansson 
\thanks{This work was supported by the Knut and Alice Wallenberg Foundation, the Swedish Strategic Research Foundation and the Swedish Research Council. The authors are with the Automatic Control Department, School of Electrical Engineering and Computer Science, KTH Royal Institute of Technology, Stockholm, Sweden. Email: {{matinj@kth.se, xinleiy@kth.se, pirani@kth.se, hsan@kth.se, kallej@kth.se}.}}}
\begin{document}
\maketitle
\thispagestyle{empty}
\pagestyle{empty}
\begin{abstract}
This paper studies the synchronization of a finite number of Kuramoto oscillators in a frequency-dependent bidirectional tree network. We assume that the coupling strength of each link in each direction is equal to the product of a common coefficient and the exogenous frequency of its corresponding head$^{1}$ oscillator. We derive a sufficient condition for the common coupling strength in order to guarantee frequency synchronization in tree networks. Moreover, we discuss the dependency of the obtained bound on both the graph structure and the way that exogenous frequencies are distributed. Further, we present an application of the obtained result by means of an event-triggered algorithm for achieving frequency synchronization in a star network assuming that the common coupling coefficient is given.
\end{abstract}
\section{Introduction}\label{sec:int}
Oscillation is the fundamental function behind the operation of many complex networks, including biological and neural networks \cite{sacre2014sensitivity,steur2009semi}. The well-celebrated Kuramoto oscillator \cite{kuramoto2012chemical} has been a paradigm for studying interconnected oscillators. Kuramoto oscillator has been originally designed to study the synchronization of coupled oscillators in chemical networks, and it has been widely used in other disciplines including synchronization in brain networks \cite{cumin2007generalising}.

From a technical point of view, synchronization of Kuramoto oscillators have been widely studied in the literature. The main results have focused on the original model of Kuramoto in a complete graph \cite{strogatz2000kuramoto} where bounds on the critical coupling are derived \cite{chopra2009exponential}, \cite{dorfler2011critical} to provide sufficient and necessary conditions for frequency synchronization. Other relevant problems have also been studied, for example synchronization of oscillators over general connected graphs \cite{jadbabaie2004stability}, synchronization with time-varying exogenous frequencies \cite{franci2010phase}, and cluster synchronization of Kuramoto oscillators in a connected, weighted and undirected graphs \cite{favaretto2017cluster}.\\[1mm]
{\it Main contributions:} This paper considers frequency synchronization of Kuramoto oscillators in a bidirectional frequency-dependent tree network. We are motivated by the interest behind studying synchronization between different areas of a complex brain-like network. Our choice of studying tree networks is encouraged by observations that large-scale inter-areal connectivity in the brain can be approximated as a tree network \cite{stam2014trees}. The idea of frequency-dependent coupling stems from the evidence of frequency-dependent synaptic coupling between neurons e.g. \cite{markram1998information}. In the presentation of the current paper, we are confined within the mathematical framework and graph theory to represent our model without a direct usage of terminologies from the domain of neuroscience.  

We consider a tree graph and assume that each link (edge) of the graph is bidirectional and its weight in each direction depends on a common coupling term as well as the exogenous frequency of the oscillator at its head$^{1}$ node \cite{xu2016synchronization}. While the oscillators' exogenous frequencies are different from each other, there is a common coupling coefficient which affects all links equally. In other words, each oscillator is connected to its neighbors with $\kappa \omega_i$, where $\kappa$ denotes the common stiffness and $\omega_i$ varies for each oscillator. We derive a sufficient condition on the bound of $\kappa$ such that the network achieves frequency synchronization. We show and discuss the dependency of the obtained bound on the exogenous frequencies as well as graph structure. 

Compared with \cite{xu2016synchronization} where star graphs with identical leaf frequencies are considered, we consider a general class of tree graphs where nodal exogenous frequencies are different from each other and also use different analytical tools from control theory.
Compared to the previous work, e.g. \cite{chopra2009exponential,jadbabaie2004stability,dorfler2011critical,favaretto2017cluster,franci2010phase}, we are considering a frequency-dependent dynamics for Kuramoto oscillators which has not been considered before. Moreover, we study synchronization in a tree graph (a non-complete graph) and derive a condition on the coupling bound which depends on the exogenous frequencies and the graph structure.

In addition to studying frequency synchronization, we present an event-based algorithm for synchronization in star networks, which is a special case of a tree network, assuming a specified $\kappa$ which may not necesarily satisfy the sufficient condition for synchronization. Compared with \cite{proskurnikov2017synchronization}, we consider a different underlying dynamics for the oscillatory network and design a different algorithm.\\[2mm]
This paper is organized as follows. Section \ref{sec:pre} presents preliminaries and problem formulation. Section \ref{sec:coup} gives a sufficient condition for the common coupling strength in order to achieve frequency synchronization. An event-triggered algorithm for synchronization in a star network is presented and analyzed in Section \ref{sec:event}. Section \ref{sec:sim} presents simulation results and Section \ref{sec:con} concludes the paper \footnote{The word {\em source} has been mistakenly used instead of {\em head} in the paper published in the proceedings of 57th IEEE Conference on Decision and Control, 2018.}\\[3mm]
\section{Preliminaries and Problem formulation}\label{sec:pre}
For a connected undirected graph $G(\mathcal V,\mathcal E)$, the node-set $\mathcal V$ corresponds to $n$ nodes and the edge-set $\mathcal E \subset \mathcal V \times \mathcal V$ corresponds to $m$ edges. The incidence matrix $B_{n \times m }$ associated to $G(\mathcal V,\mathcal E)$ describes which nodes are coupled by an edge. Each element of $B$ is defined as follows
\begin{equation*}
b_{i \ell} = 
\begin{cases}
+1 & \text{if node $i$ is at the positive end of edge $\ell$} \\
-1 & \text{if node $i$ is at the negative end of edge $\ell$}\\
0 & \text{otherwise},
\end{cases}
\end{equation*}
where the labeling of the nodes can be done in an arbitrary fashion. The matrix $L= B B^T$ is called the graph Laplacian and $L_g=B^T B$ is the edge-Laplacian. If the underlying graph is connected, the eigenvalues of the Laplacian matrix are ordered as $0=\lambda_1(L) < \lambda_2(L)\leq ... \leq \lambda_n(L)$, where $\lambda_2(L)$ is called the {\it algebraic connectivity} of the network. If the undirected graph is a tree, all eigenvalues of $L_g$ (which are $n-1$ eigenvalues) are equal to the nonzero eigenvalues of $L$ \cite{mesbahi2010graph}. In this paper, we consider  trees which are a  subclass of connected graphs without cycles, i.e., any two nodes are connected by exactly one unique path. The edge Laplacian of a tree graph is invertible \cite{mesbahi2010graph}. 
\begin{definition}
A subset $S \subset \R^n$ is said to be forward invariant with respect to the differential equation $\dot x=f(x)$ provided that each solution $x(\cdot)$ with $x(0) \in S$ has the property that $x(t) \in S$ for all positive $t$ in the domain of definition of $x(\cdot)$ \cite{sontag2001structure}.
\end{definition}
{\bf{Notation}}\\
Symbol $\mathbf{1}_n$ is a $n$-dimensional vector and $\mathbf{1}_{m \times m}$ represents a $m \times m$ matrix whose elements are all equal to $1$. The notation $x_{i,j}$ is equivalently used for $x_i-x_j$. The notation $\theta_{i,k}$ indicates that the node $i$ of graph $G$ is connected to the edge $k$. The minimum eigenvalue of the positive definite matrix $M$ is denoted by $\lambda_1(M)$.  
\subsection{Problem formulation}\label{sec:pf}
Consider $n$ oscillators communicating over a connected and bidirectional graph. The original Kuramoto model follows
\begin{equation}\label{eq:kh}
\dot{\theta_i}= \omega_i - \kappa \sum_{j \in {\cal N}_i} \sin(\theta_{i}-\theta_{j}),
\end{equation}
where $\theta_i \in \R$, $\omega_i>0$ are the phases and exogenous frequencies of oscillator $i$, and ${\cal N}_i$ denotes the set of neighboring nodes of node $i$. The parameter $\kappa > 0, \kappa \in \R$ is the constant coefficient of the coupling strength of all links of the graph. 
We now continue with a different model where the dynamics of each node follows \cite{xu2016synchronization} 
\begin{equation}\label{eq:m1h}
\dot{\theta_i}= \omega_i - \kappa \omega_i \sum_{j \in {\cal N}_i} \sin(\theta_i-\theta_j). 
\end{equation}
This model can be interpreted as a bidirectional communication where the weights of coupling of each edge at each direction depends on the frequency of the head node. This makes the interaction topology a directed and weighted graph. The model of the network in compact form is
\begin{equation}\label{eq:cl}
\dot{\boldsymbol \theta}= \boldsymbol \omega (\mathbf{1}_n - \kappa B \sin(B^T \boldsymbol \theta)),
\end{equation}
where $\boldsymbol \omega \triangleq \diag(\omega_1, \omega_2, ..., \omega_n)$ is a diagonal matrix such that $\omega_i >0$ represents the exogenous frequency of node $i$, $\boldsymbol\theta \triangleq [\theta_1, \theta_2, ..., \theta_n]^T$, and $\sin$ function acts element-wise. 
As presented in \cite{dorfler2011critical}, the notions of synchronization include
\begin{itemize}
    \item phase cohesiveness, \ie,\  $|\theta_i-\theta_j| \leq \eta$,
    \item phase-synchronization, \ie,\  $\theta_i=\theta_j, \forall i,j$,
    \item frequency synchronization, \ie,\  ${\dot\theta}_i={\dot\theta}_j, \forall i,j$.
\end{itemize}
In this paper we first characterize the sufficient condition for the coupling strength such that for $\kappa > \Delta$, the phase cohesiveness and frequency synchronization occur provided that the initial conditions of the oscillators are within a prescribed bound. Second, we present an event-triggered mechanism to achieve frequency synchronization in a star network. 
\section{The coupling strength}\label{sec:coup}
This section studies a sufficient condition for the coupling strength $\kappa$ to guarantee phase-cohesiveness and frequency synchronization for the network presented by \eqref{eq:cl}. 
\begin{assumption}\label{ass0}
The communication topology for the network with node dynamics in \eqref{eq:m1h} is a tree.
\end{assumption}
\begin{assumption}\label{ass1}
The initial relative phase $\theta_{i}(0)-\theta_{j}(0) \in [-\eta,\eta], j \in {\cal N}_i$, where $\eta>0, \eta=\frac{\pi}{2}-\eps$ for some $\eps>0$.
\end{assumption}
\begin{assumption}\label{ass2}
All exogenous frequencies $\omega_i$ are strictly positive, i.e.,\ $\omega_i \geq \zeta>0$.
\end{assumption}
Since we are interested in frequency synchronization, \ie\ ${\dot\theta}_i={\dot\theta}_j$, and $|{\theta}_i-{\theta}_j| \leq \eta$, let us first write the compact relative phase dynamics, $B^T \dot{\boldsymbol \theta}$, by multiplying \eqref{eq:cl} with $B^T$ as follows
\begin{equation}\label{eq:rcl}
B^T \dot{\boldsymbol \theta}= B^T \boldsymbol \omega \mathbf{1}_n - \kappa\ B^T \boldsymbol \omega B \sin(B^T \boldsymbol \theta).
\end{equation}
\begin{lemma}\label{lem1}
Under Assumptions \ref{ass0} and \ref{ass2}, the matrix $B^T \boldsymbol \omega B$ is positive definite and its smallest eigenvalue is positive, i.e., $\lambda_{1}(B^T \boldsymbol \omega B) > 0$.
\end{lemma}
\begin{proof}
Consider $x^T B^T \boldsymbol \omega B x$, $x \in \R^n$ and define $y= B x$. We have, 
$$x^T B^T \boldsymbol \omega B x= y^T \boldsymbol \omega y.$$
Since $\boldsymbol \omega >0$, $y^T \omega y \geq \omega_{\min} y^T y$. Thus,
$$x^T B^T \boldsymbol \omega B x \geq \omega_{\min} y^T y \geq 0.$$
Now $y^T y=0$ if and only if $y=\mathbf{0}$.  If we show that $y=\mathbf{0}$ implies $x=\mathbf{0}$, then $B^T \boldsymbol \omega B >0$. We argue as follows. Assume that $y=B x=\mathbf{0}$. Then, $B^T B x=\mathbf{0}$. Since for a tree graph, $B^T B$ is invertible, we conclude that $x=\mathbf{0}$ which ends the proof. 
\end{proof}
Recall Assumption \ref{ass1} and for a given $\eps >0$ define 
\be\label{eq:s}
S=\{\boldsymbol \theta \in \R^n: |\theta_{i}-\theta_{j}| \leq \eta,  \eta=\frac{\pi}{2}-\eps, \forall (i,j) \in \mathcal E\}.
\ee
In addition, take $V(\boldsymbol\theta)= 2 \sin^T (\frac{B^T \boldsymbol \theta}{2}) \sin(\frac{B^T \boldsymbol \theta}{2})$ and define
\be\label{eq:s2}
S'=\{ {\boldsymbol \theta} \in \R^n: V(\boldsymbol \theta) \leq c(\eps)\}
\ee
such that $S' \subset S$ is the largest level set of $V(\boldsymbol\theta)$ that fits in $S$. Notice that $c(\eps)>0 \in \R$ is defined for the given $\eps$.     
\begin{proposition}\label{pr1}
For the relative angle dynamics in \eqref{eq:rcl} under Assumptions \ref{ass0}-\ref{ass2}, consider the set $S'$ in \eqref{eq:s2} and the set $S$ in \eqref{eq:s}. Then $S' \subset S$ is forward invariant for system \eqref{eq:rcl} provided that $\kappa \geq \frac{|\Delta^{\omega}_{\max}|}{\lambda_{1}(B^T \boldsymbol \omega B) \cos(\eps)}$, where $\Delta^{\omega}_{\max}=\max_{(i,j)\in \mathcal{E}}{|\omega_i - \omega_j|}$. 
\end{proposition}
\begin{proof}
Consider $V= 2 \sin^T (\frac{B^T \boldsymbol \theta}{2}) \sin(\frac{B^T \boldsymbol \theta}{2})$ as the Lyapunov function and define $S'$ as in \eqref{eq:s2}. Write $V$ in the element-wise form, $V= 2 \sum_{i} \sum_{j} \sin^2(\frac{{\theta}_{i}-\theta_{j}}{2})$, and calculate $\dot V$. We obtain
\be\begin{aligned}\label{eq:p21}
\dot{V}&= \sum_{i}\sum_{j} 2 \sin(\frac{{\theta}_{i}-{\theta}_{j}}{2}) \cos(\frac{{\theta}_{i}-{\theta}_{j}}{2}) ({\dot\theta}_{i}-{\dot\theta}_{j})\\
&= \sum_{i}\sum_{j} \sin({\theta}_{i}-\theta_{j}) ({\dot\theta}_{i}-{\dot\theta}_{j})\\
&= \sin^T(B^{T} \boldsymbol \theta)\  B^T \dot{\boldsymbol \theta}.
\end{aligned}\ee
Replacing \eqref{eq:rcl} in the above gives 
\be\begin{aligned}\label{eq:p22}
\dot{V}&=\sin^T(B^T \boldsymbol \theta) B^T (\boldsymbol \omega \bold{1}_n - \kappa \boldsymbol \omega B \sin(B^T \boldsymbol \theta))\\
&= \sin^T(B^T \boldsymbol \theta) B^T \boldsymbol \omega \bold{1}_n - \kappa \sin^T(B^T \boldsymbol \theta) B^T \boldsymbol \omega B \sin(B^T \boldsymbol \theta).
\end{aligned}\ee
Notice that $B^T \boldsymbol \omega \bold{1}_n \leq \Delta^{\omega}_{\max} \bold{1}_m$ where $\Delta^{\omega}_{\max}$ denotes the maximum of $|\omega_i-\omega_j|$ and $m=n-1$ is the number of edges of a tree graph. Define $y=\sin(B^T \boldsymbol \theta)$. Hence, $$-|\Delta^{\omega}_{\max} y^T| \bold{1}_m \leq y^T B^T \boldsymbol \omega \bold{1}_n \leq |\Delta^{\omega}_{\max} y^T| \bold{1}_m ,$$
where $|\Delta^{\omega}_{\max} y^T|= |\Delta^{\omega}_{\max}| | y^T|$ and $| y^T|$ is the element-wise absolute value. Also, from Lemma \ref{lem1}, we have that $\lambda_{1}(B^T \boldsymbol \omega B) > 0$. Applying the above simplifications in \eqref{eq:p22}, we obtain
\be\begin{aligned}\label{eq:p23}
\dot{V}&\leq |\Delta^{\omega}_{\max} y^T| \bold{1}_m - \kappa \lambda_{1}(B^T \boldsymbol \omega B) y^T y\\ 
&\leq -(\alpha |y|- \beta \bold{1}_m)^T (\alpha |y|- \beta \bold{1}_m) + m \beta^2 \\
&\leq \sum_{k=1}^{m} -((\alpha |y_k| -\beta)^2 - \beta^2),
\end{aligned}\ee  
where $\alpha=\sqrt{\kappa \lambda_{1}(B^T \boldsymbol \omega B)}$, $\beta=\frac{\Delta^{\omega}_{\max}}{2\sqrt{\kappa \lambda_{1}(B^T \boldsymbol \omega B)}}$, and $|y_k|=|\sin (\theta_{i,k}-\theta_{j,k})|$ with $k$ denoting the edge connecting two nodes $i$ and $j$. Define ${\dot V}_k = -((\alpha |y_k| -\beta)^2 - \beta^2)$. Based on \eqref{eq:p23}, we have ${\dot V}_k \leq 0$ holds if $|y_k|=0$ or
\be\label{eq:k}
|y_k|=|\sin (\theta_{i,k}-\theta_{j,k})| > \frac{2 \beta}{\alpha}, \beta>0, \alpha>0,
\ee
that is $\kappa \geq \frac{|\Delta^{\omega}_{\max}|}{\lambda_{1}(B^T \boldsymbol \omega B) |\sin(\theta_{i,k}-\theta_{j,k})|}$.\\
\begin{figure}[h]
\centering
\includegraphics[scale=0.65]{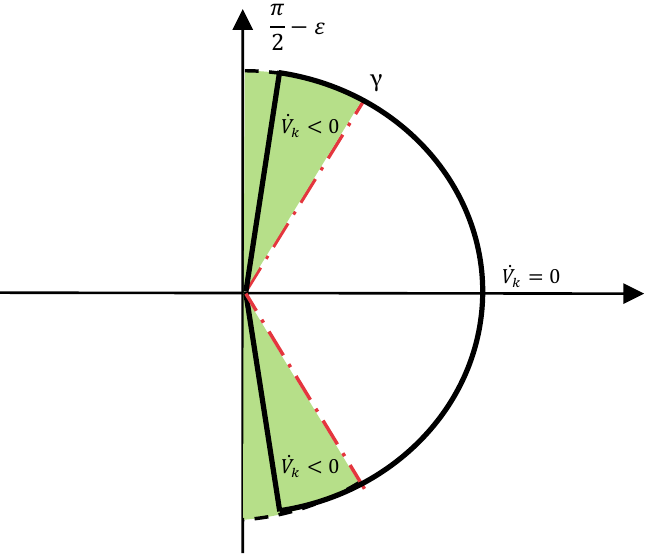}
\caption{The sign of ${\dot V}_k$ for $\theta_{i,j} \in [-\frac{\pi}{2}+\eps,\frac{\pi}{2}-\eps]$.}\label{fig:ly}
\end{figure}
Since $|\sin(\theta_{i,k}-\theta_{j,k})|$ can be very small then there is a region where ${\dot V}_k >0$ (see Figure \ref{fig:ly}). Hence, if $\kappa$ is larger than 
$\kappa \geq \frac{|\Delta^{\omega}_{\max}|}{\lambda_{1}(B^T \boldsymbol \omega B) |\sin \gamma|}$, ${\dot V}_k$ will be negative on the set $\gamma <|\theta_{i,j}|<\frac{\pi}{2}$. 
Now, if we take $\gamma$ close enough to $\frac{\pi}{2}-\eps$, the sufficient condition for $S'$ to be forward invariant is that
\be\label{eq:keps}
\kappa \geq \frac{|\Delta^{\omega}_{\max}|}{\lambda_{1}(B^T \boldsymbol \omega B) \cos(\eps)},
\ee 
which ends the proof.\\[2mm]
\end{proof}
For a tree graph, the symmetric matrix $A=B^T\boldsymbol\omega B$, which is in the form of a weighted edge Laplacian, has the following structure 
\be\label{eq:topo}
|A_{m \times m}|=
  \begin{bmatrix}
    |\omega_{i,1}+\omega_{j,1}| & |\omega_\ell^{1,2}| & \ldots & |\omega_\ell^{1,m}|\\
    \vdots & \vdots & \vdots & \vdots\\
    |\omega_\ell^{m,1}| & |\omega_\ell^{m,2}| & \ldots & |\omega_{i,m}+\omega_{j,m}|\\
  \end{bmatrix},
\ee
where $\omega_{i,k}$, $\omega_\ell^{k,p}$ denote the frequency of node $i$ connected to the link $k$ and the frequency of the shared node $\ell$ of two links $k$, and $p$, respectively. The following proposition provides bounds on $\lambda_{1}(B^T \boldsymbol \omega B)$ based on the network topology and exogenous frequencies which can be used in \eqref{eq:keps}.
\begin{proposition}\label{pr2}
The minimum eigenvalue of $B^T \boldsymbol \omega B$, $\lambda_{1}(B^T \boldsymbol \omega B)$, for a tree structure is lower bounded by 
\begin{align}\label{eq:threebounds}
\max\bigg\{\underbrace{\omega_{\min} \lambda_{2}(B B^T)}_{(i)}, \underbrace{\min_{k \in \mathcal{E}}\{(2-d_i)\omega_{i,k}+(2-d_j)\omega_{j,k}\}}_{(ii)}\bigg\}\nonumber\\
\leq \lambda_{1}(B^T \boldsymbol \omega B)\leq \min_{k \in \mathcal{E}} \{\omega_{i,k}+\omega_{j,k}\},
\end{align}
where $d_i$ is the degree of node $i$ of the underlying graph.
\end{proposition}
\begin{proof}
We motivate the reasoning behind each of the three elements in \eqref{eq:threebounds}. Recall that based on Lemma \ref{lem1},  $x^T B^T \boldsymbol \omega B x \geq a x^T x$ where $a>0$.\\ 
i) We have $\omega_{\min} I_n \leq \boldsymbol\omega \leq \omega_{\max} I_n$. Thus, $B^T \boldsymbol \omega B> \omega_{\min} B^T B > \omega_{\min} \lambda_{1}(B^T B) I_n$. Since for a tree graph $B^T B$ is invertible \cite{mesbahi2010graph}, we have $B^T B > \lambda_{2}(L) I_n$ which gives the first bound that is rather conservative.\\
ii) The second bound is obtained as a lower bound on the smallest singular value of a general symmetric matrix $A_{m \times m}$ 
\cite{johnson1989gersgorin}
$$\sigma_{\rm min}(A)\geq \min_{i}\{|a_{i,i}| - \left(\sum_{j=1, j \neq i}^{m} |a_{i,j}|\right)\}.
$$
For the case of $A=B^T\boldsymbol\omega B$ as in \eqref{eq:topo}, the above bound will be a lower bound for its smallest eigenvalue. We have $|a_{kk}|=\omega_{i,k}+\omega_{j,k}$ where $k$ is the edge which connects node $i$ to node $j$. Also, $\sum_{j=1, j \neq i}^{m} |a_{i,j}|= 2(d_i-1)\omega_{i,k}+ 2(d_j-1)\omega_{j,k}$, and it gives the result.

The upper bound comes from the Rayleigh quotient inequality \cite{Horn}  is obtained by choosing $\boldsymbol x=\mathbf{e}_k$ where  $\mathbf{e}_k$ is the $k$-th vector of the canonical basis and
$k=\arg \min_{k \in \mathcal{E}} \{\omega_{i,k}+\omega_{j,k}\}$. 
\end{proof}
A comparison between the tightness of the bounds proposed in \eqref{eq:threebounds} is discussed in Section \ref{sec:sub}. 
\begin{proposition}\label{pr3}
Under Assumptions \ref{ass0}-\ref{ass2}, the network in \eqref{eq:rcl} achieves frequency synchronization, \ie\  ${\dot\theta}_i={\dot\theta}_j, \forall i,j \in \mathcal V$ if the condition on $\kappa$ in \eqref{eq:keps} holds.
\end{proposition}
\begin{proof}
Consider the relative phase dynamics as in \eqref{eq:rcl}. Define $z=B^T \dot{\boldsymbol \theta}$. Calculating $\dot z$, we obtain
\be\label{eq:net22}
\dot z= - \kappa A W^{\cos\theta} z,
\ee
where $A= B^T \boldsymbol\omega B$ and $W^{\cos\theta}$ is a $m \times m$ diagonal matrix with diagonal elements equal to $\cos(\theta_{i,k}-\theta_{j,k})$ where $k \in \{1,\ldots,m\}$ denotes the index of edge $k$ of the graph. Consider the Lyapunov function $V= z^T A^{-1} z$. Assume that the condition on $\kappa$ in \eqref{eq:keps} holds. Thus, from proposition \ref{pr1}, we have $\cos(\theta_{i,k}-\theta_{j,k}) \geq \sin(\eps) >0 $. Hence, $\sin(\eps) I_m \leq W^{\cos\theta} \leq I_m$. We obtain ${\dot V} \leq - \kappa \sin(\eps) z^T z$. Thus, the system in \eqref{eq:net22} is uniformly asymptotically stable which ends the proof. 
\end{proof}

\begin{remark}\label{rem1}
Notice that if $\boldsymbol \omega= \omega I_m, \omega \in (0, +\infty)$, then both the phase and frequency synchronization will be achieved $\forall\  \kappa >0$ \cite{jadbabaie2004stability}. This case is identical with the original Kuramoto model (as in \eqref{eq:kh}) with identical frequencies. 
\end{remark}
\subsection{Examples: Effects of the graph structure, size and distrubution of exogenous frequencies on the sufficient $\kappa$}\label{sec:sub}
This section first presents some examples to show that not only the magnitudes of exogenous frequencies affect the sufficient bound of $\kappa$ but also the way that these frequencies are distributed. Moreover, we provide some examples to compare the tightness of the bounds proposed in \eqref{eq:threebounds} with respect to the graph structure and size. 
\begin{example}\textbf{(Effects of the distribution of exogenous frequencies)}:
We provide some examples to show that not only the magnitude of $|\Delta^{\omega}_{\max}|$ affects $\kappa$ but also the way that the exogenous frequencies are distributed. We consider the effect of changing a single exogenous frequency $\omega_i$ on $\lambda_1(B^T\boldsymbol\omega B)$ and consequently on \eqref{eq:keps}. Suppose that the frequencies of all nodes are 10, except one which is 1. We want to assign this frequency $\omega=1$ to one of the nodes in the network such that the resulting eigenvalue $\lambda_1(B^T\boldsymbol\omega B)$ is maximized. For the case of a star graph, if we place it in the center (hub), $\lambda_1(B^T\boldsymbol\omega B)=10$ is obtained and if we put it on one of the leaves, it gives $\lambda_1(B^T\boldsymbol\omega B)= 2.16$. For the graph shown in Fig. \ref{fig:aafwdnet}-b, which consists of two stars with the same size connected via a single edge, due to the symmetry there are two possibilities; either $\omega_1=1$ or $\omega_8=1$. For the former case, $\lambda_1(B^T\boldsymbol\omega B)=2.67$ is obtained, while the latter case gives $\lambda_1(B^T\boldsymbol\omega B)=1.54$. For the line graph shown in Fig. \ref{fig:aafwdnet}-c, the largest value for $\lambda_1(B^T\boldsymbol\omega B)$ takes place when we assign $\omega_4=1$, which gives $\lambda_1(B^T\boldsymbol\omega B)=1.98$ and it decreases by going from the center of the line to one the ends which gives $\lambda_1(B^T\boldsymbol\omega B)=0.86$.

We observe that in the above examples the optimal place of the vulnerable node should be one of the well-known centralities. Notice that the values of $\lambda_1$ in these examples are reported based on the exactly calculated $\lambda_1(B^T\boldsymbol\omega B)$ for the given graphs. Using the estimated value of $\lambda_1$ in \eqref{eq:threebounds} leads to a similar conclusion. For example, for the star graph in Fig. \ref{fig:aafwdnet}-a, the bound in \eqref{eq:threebounds} gives $5 \leq \lambda_1 \leq 11$ if $\omega_1=1$ and $1 \leq \lambda_1 \leq 11$ if the exogenous frequency of one of leaves is equal to one. 

\begin{figure}[t]
\centering
\includegraphics[scale=0.5]{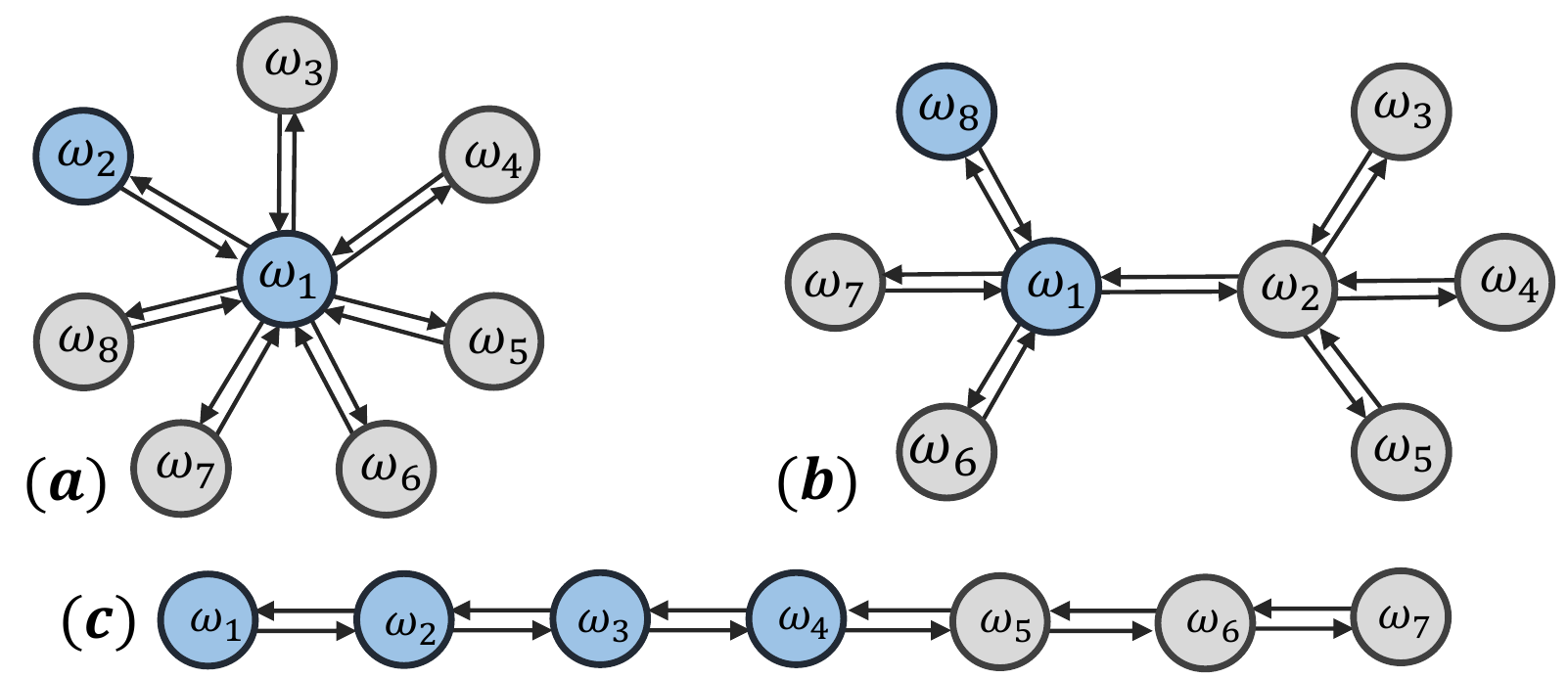}
\caption{Different locations of the vulnerable node in three specific graphs.}\label{fig:aafwdnet}
\end{figure}
\end{example}

\begin{example}\textbf{(Comparing bounds in \eqref{eq:threebounds} and the effect of network size and structure)}:
For the graph shown in Fig. \ref{fig:awdnet}-a the lower bound (i) gives 1, while bound (ii) is 99. Hence, bound (ii) is tighter. If we keep increasing the number of leaves up to 100, the lower bound (ii) is still tighter than (i) (for Fig. \ref{fig:awdnet}-b, (i) gives 1 and (ii) gives 2). For graph Fig. \ref{fig:awdnet}-c,  bound (i) gives 1 and bound (ii) gives  $-49$.  Hence, depending on the network structure, either of the two lower bounds become tighter.

We should note that the largest algebraic connectivity among all trees belongs to star graphs, which is 1. For most of the tree structures, $\lambda_2(L)$ scales with the size of the network, e.g., line graphs. Fig. \ref{fig:awdnet}-d shows the role of network size on the scaling of bound (i). In this example bound (ii) gives zero. The graph topology is a line graph and we know that for these graphs $\lambda_2(L)=1-\cos (\frac{\pi}{n})$ \cite{Abreu}. Thus, bound (i) gives $\omega_{\rm min}\left(1-\cos (\frac{\pi}{n})\right)$ which bigger than zero, although it goes to zero as the network size grows. 

\begin{figure}[b]
\centering
\includegraphics[scale=0.43]{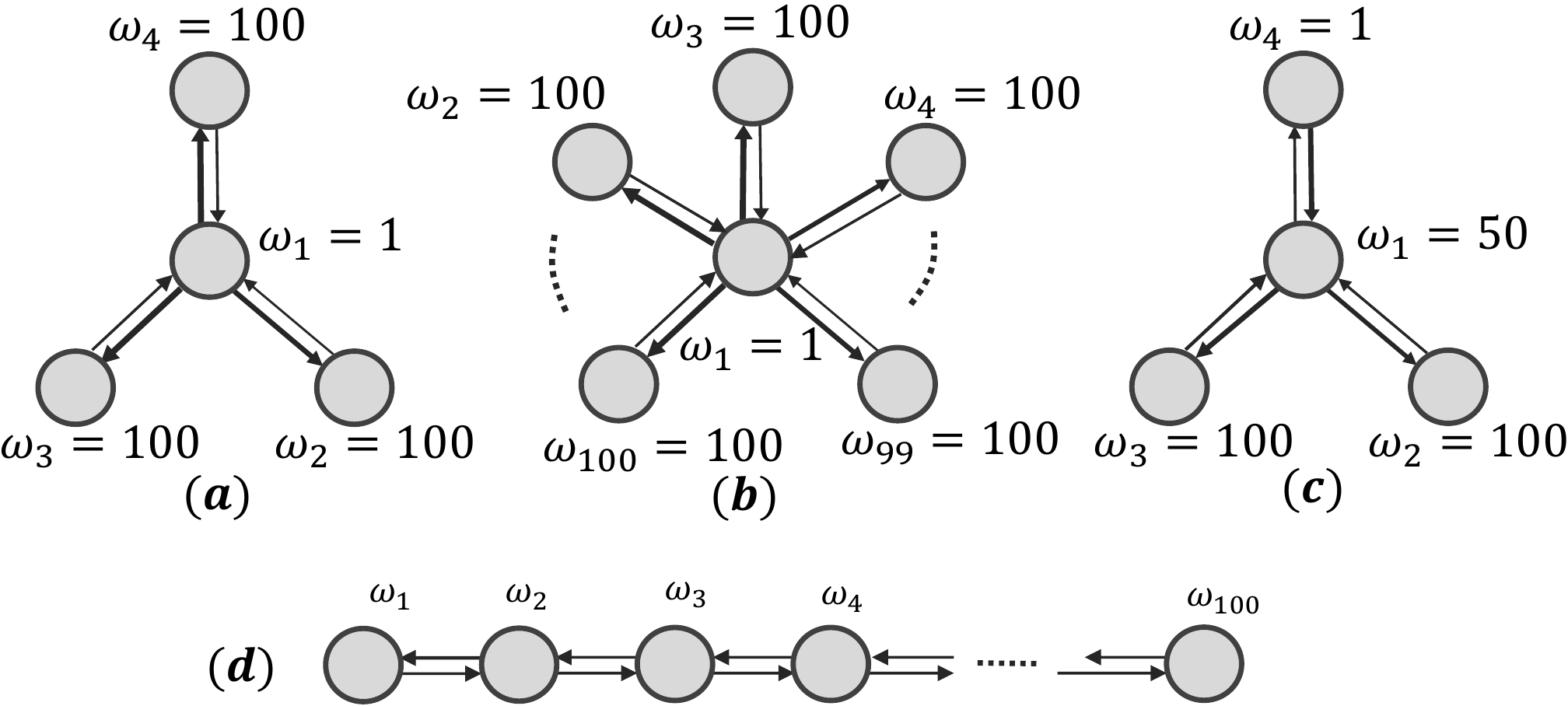}
\caption{Examples which show the tightness of bounds proposed in \eqref{eq:threebounds} and the role of the network scaling and structure on each lower bound.}\label{fig:awdnet}
\end{figure}
\end{example}
\section{Application: An event-based algorithm for frequency synchronization in a star network}\label{sec:event}
In this section, we present an application of using the sufficient coupling strength obtained in Proposition \ref{pr1} by means of a centralized event-based algorithm (e.g. \cite{dimarogonas2012distributed}) for synchronization in a star network. Following the previous section, we are interested in frequency synchronization and phase-cohesiveness, \ie\ $|\theta_j-\theta_i| < \eta$ and ${\dot{\theta}}_i={\dot{\theta}}_j$. To this purpose, we assume that the coupling stiffness $\kappa$ is given. Hence, to derive the network towards frequency synchronization, we manipulate $\omega_i$ using the results of the previous section such that for {\bf each} edge $|\theta_i-\theta_j| < \eta$ ($\eta=\frac{\pi}{2}-\eps, \eps>0$) is enforced. 
Consider a star graph (see Fig. \ref{fig:aafwdnet}-a) with the node dynamics as follows
\be\begin{array}{cc}\label{eq:e1-1}
\dot{\theta}_h = \omega_h  (1- \kappa \sum_{i} \sin(\theta_h-\theta_i)),\\[2mm]
\hspace{5mm} \dot{\theta}_i = (\omega_i + \alpha_i) (1+\kappa \sin(\theta_h-\theta_i)),
\end{array}\ee
where the central node of the star is called {\em hub}, denoted by ${\theta}_h$, and other nodes are called {\em leaf}, denoted by ${\theta}_i$. The problem is how to design $\alpha_i$ in order to achieve our goal. One approach could be based on designing an $\alpha$ controller to continuouesly regulate $\omega_i$, e.g. \cite{jafarian2016disturbance}. In this paper, however, we are interested in an event-based approach which does not require updating the control action for all times.
\begin{assumption}\label{assini}
We assume that
\begin{enumerate}
\item the exogenous frequencies $\omega_j$ with $j \in \{h,i\}$ are slow enough to be estimated as constant for a large enough period of time,
\item relative phases and their derivatives, \ie,\  $\theta_i-\theta_h$ and ${\dot{\theta}}_i-{\dot{\theta}}_h$, are known to the hub.
\end{enumerate}
\end{assumption}
Under the above assumption, we design a centralized event-triggered algorithm such that the hub updates $\boldsymbol \alpha =(\alpha_2,\ldots,\alpha_n)^T$ for achieving frequency synchronization. For system \eqref{eq:e1-1}, we define a set of triggering times $t_h^0, t_h^1, t_h^2, \ldots$ at which the vector $\boldsymbol \alpha$ gets updated, such that $\boldsymbol \alpha(t)=\boldsymbol \alpha(t_h^k), t \in [t_h^k,t_h^{k+1})$.\\
Before presenting the algorithm, we first define the events ($E_1, E_2$) and the required calculations for updating $\boldsymbol \alpha$.\\
Event $E_1^i$ is activated if the relative phase of nodes $i$ and $h$ is larger than a prescribed limit. In this case, the leaf exo-frequency will be updated by adjusting $\alpha_i$ while the hub exo-frequency is kept unchanged. Event $E_1$ is activated if there exists one edge which meets the triggering condition. Hence, $E_1= \bigcup\limits_{i=2}^{n} E_1^i$, where $E_1^i$ denotes event $E_1$ for node $i$.\\
Besides event $E_1$, event $E_2$ is designed to update the triggering condition of each edge in order to avoid chattering (repetitive switchings) of $E_1^i$ \cite{jafarian2017robust} (see Remark \ref{chat}). We write $E_2= \bigcup\limits_{i=2}^{n} E_2^i$.\\
{\bf Event $E_1$ (Update of $\alpha_i$):} $\exists i \in \{2, \ldots, n\} \quad \text {s.t.} \quad |\theta_h-\theta_i| > \eta$, with $\eta=\frac{\pi}{2}-\eps >0$. Let $\alpha_i$ denote $\alpha_i(t_h^k)$ and $\alpha_i^+$ denote $\alpha_i(t_h^{k+1})$ for node $i$. 

The update of $\alpha_i$ (\ie,\ $\alpha_i^+$) should guarantee that the link $h,i$ is contributing to the decrease of the overall Lyapunov function of the system (see the proof of Proposition \ref{pr1}). Hence, the sufficient condition on $\kappa$ should locally hold. Considering \eqref{eq:keps}, $\alpha_i^+$ should locally guarantee that $$\kappa > \frac{|{\omega}_h-(\omega_i+\alpha_i^+)|}{\lambda_{1}(B^T \boldsymbol \omega B) \cos(2 \eps)}.$$ In the above, considering $\cos(2 \eps)$ guarantees that $\dot V$ will be locally negative for $\theta_{h,i}> \frac{\pi}{2}- 2 \eps$, and thus it will negative for ${\theta}_{h,i} > \eta,\ \eta=\frac{\pi}{2}-\eps$ (see Fig \ref{fig:ly}). To locally estimate $\lambda_{1}(B^T \boldsymbol \omega B)$, which is required for updating $\alpha_i$, we use the estimation based on $\omega_{\min} \lambda_{2}({B B^T})$ (see in \eqref{eq:threebounds}). Define $\omega_{\min}^L=\min \{{\omega}_h, \omega^\ast\}$, where $\omega_{\min}^L$ denotes the local estimation of $\omega_{\min}$ and $\omega^\ast$ is the desired value for ${\omega_i+\alpha_i^+}$. Since $\omega^\ast$ is a design choice, we opt for the case where $\omega^\ast > {\omega}_h$ (motivated by the examples in Section \ref{sec:sub}). Notice that for a star graph $\lambda_{2} (B B^T)=1$. Hence, from \eqref{eq:keps}, $\kappa > \frac{\omega^\ast-{\omega}_h}{{\omega}_h \cos(2 \eps)}$ should hold. The hub then calculates
\begin{equation}\label{eq:e3es2}
\omega^\ast= {\omega}_h (1+\Delta \cos(2 \eps)),  
\end{equation} 
and updates 
\begin{equation}\label{eq:ai}
\alpha_i^+ = \omega^\ast- (\omega_i+\alpha_i).
\end{equation}
Notice that in calculation of $\omega^\ast$, it is assumed that $\omega_h$ is known to the hub. In fact, the structure of the star graph together with Assumption \ref{assini}-2 allow the hub to calculate both $\omega_h$ and $\omega_i$ by using the first and second derivatives of $\theta_h-\theta_i$.
Now, we present our event-triggered algorithm as follows.\\[1mm]
Consider Assumption \ref{assini}. Let $t_0\ge0$ denote the initial time and $\{t^k_h\}_{k=1}^{\infty}$ denote the triggering times of the hub determined by E1 or E2 (see Remark \ref{chat}).\\[3mm] 
\noindent {\it Algorithm 1}:
\begin{algorithmic}[1]
\State Choose $\eta>0$ and $\eps>0$.
\State Initialize $\alpha_i=0$, $t^{0}_h=t_0$ and $k=0$.
\State From time $s=t^{k}_h$, the hub continuously senses $\theta_h-\theta_i$ to detect $E=E_1 \bigcup E_2$, where $E_1= \bigcup\limits_{i=2}^{n} E_1^i$ and $E_2= \bigcup\limits_{i=2}^{n} E_2^i$, such that $\tau=\inf\{r\ge s:~\exists i\ \text{s.t.}\  E_1^i\  \text{or}\  E_2^i\},$ where $E_1^i$ implies $\{~\exists i,\ ~|\theta_h-\theta_i|>\eta\}$ and $E_2^i$ implies $\{~\exists i,\ ~|\theta_h-\theta_i|<\eta-\eps\}$.
\State If $E$, the hub determines $t^{k+1}_h=\tau$, and for every $i$ which meets $E_1^i$, the hub updates $\alpha_i$ based on \eqref{eq:ai} and replaces the definition of $E_1^i$ with $E_2^i$. Also, for every $j$ which meets $E_2^j$, the hub replaces the definition of $E_2^j$ with $E_1^j$. The hub goes back to Step 3.
\end{algorithmic}

\begin{remark}\label{chat}[Avoiding chattering]
After triggering $E_1^i$, the same event can be immediately triggered since there will be a $\Delta t$ till the condition of event $E_1^i$ is violated. To prevent this behavior, event $E_2^i$ is introduced. In fact, $E_2^i$ will temporarily replace $E_1^i$ and thus will avoid chattering of this event. Notice that although it is possible that two different edges trigger an event at the same time (for example edge $i$ triggers $E_1^i$ and edge $j$ triggers $E_2^j$ simultaneously), it is impossible that $E_1^i,E_2^i$ occur for the same edge simultaneously.
\end{remark}
System \eqref{eq:e1-1} with the above event-triggered algorithm can be represented as the following hybrid system with state $(\boldsymbol\theta,\boldsymbol\alpha)$ such that the continuous evolution of the system obeys 
\begin{equation}\label{eq:e1}
\begin{array}{lll}
\dot{\boldsymbol \theta} &= \boldsymbol \omega^{\alpha} (\mathbf{1}_n - \kappa B \sin(B^T \boldsymbol \theta)),\\[1mm]
\dot {\boldsymbol \alpha} &= 0,
\end{array}\end{equation}
where $\boldsymbol \omega^{\alpha}$ is a diagonal matrix whose diagonal elements are $\omega_1^\alpha=\omega_h$ and $\omega_i^\alpha=\omega_i+\alpha_i$. Also, if there is a link which meets the jump condition, the following discrete transition occurs
\begin{equation}\label{eq:e2}
{\boldsymbol \theta}^+={\boldsymbol \theta}, \quad \alpha_i^+ = \omega^\ast-(\omega_i+\alpha_i).  
\end{equation}
\begin{proposition}\label{pr4}
There is a lower bound on the inter-triggering times of the solutions to hybrid system \eqref{eq:e1}-\eqref{eq:e2} considering a star graph with Assumptions \ref{ass1}-\ref{assini}. Moreover, the set $S' \subset S$, with $S'$ in \eqref{eq:s2} and $S$ in \eqref{eq:s} , is forward invariant for system \eqref{eq:e1}-\eqref{eq:e2}. 
\end{proposition}
\begin{proof}
First we prove that $S'$ is forward invariant. Since at the switches the state $\theta$ stays unchanged, we take a similar Lyapunov function as in the proof of Proposition \ref{pr1}. Since event $E_1$ prevents the phase differences to grow larger than $\eta$, and also the correcting terms $\alpha_i$ are designed such that the sufficient condition on $\kappa$ is respected, then based on the argument in the proof of Proposition \ref{pr1}, the set $S'$ is forward invariant for \eqref{eq:e1}-\eqref{eq:e2}. To prove the existence of a non-zero dwell time, we argue that an edge cannot trigger $E_1^i$ and $E_2^i$ at the same time (see Remark \ref{chat}). Also, after triggering either $E_1^i$, at least a distance ($\eps$) should be paved with a bounded velocity (since $\dot {\boldsymbol \theta}$ is bounded) which implies $E_2^i$ does not instantaneously happen after $E_1^i$. The latter  indicates that there exists a non-zero dwell time between each two triggering times. Notice that at each triggering time, it is possible to have more than one edge that trigger an event but number of nodes if finite. Hence, there is no zeno behavior and the solutions can evolve in time.
\end{proof}
\begin{proposition}\label{pr5}
Under Assumptions \ref{ass1}-\ref{assini}, a star network with the dynamics as in \eqref{eq:e1}-\eqref{eq:e2} with the event-triggered Algorithm 1 achieves frequency synchronization.
\end{proposition}
\begin{proof}
Define $z=B^T \dot{\boldsymbol \theta}$. Following \eqref{eq:net22}, consider
\be\label{eq:net22h}
\dot z= - \kappa A^\alpha W^{\cos\theta} z,
\ee
where $A^\alpha= B^T \boldsymbol\omega^{\sigma} B$ and $W^{\cos\theta}$ is defined as before. We now argue that there exists a time $T$ at which all nodal exogenous frequencies are either updated or will stay unchanged. We reason as follows. If $\theta_h-\theta_i$ exceeds the designed upper-bound, an event will be generated and the exogenous frequency of node $i$ will be replaced by $\omega^\ast$. If no event is generated by node $i$, we conclude that $\theta_h-\theta_i$ is small enough and within the desired bound. Since, the number of nodes are finite, there exits a finite time $T$ at which the exogenous frequencies will not get updated anymore. We label the exogenous nodal frequency of node $i$ after time $T$ by ${\omega_i}^\ast$. Take $V_2=z^T M^{-1} z$ with $M= \omega_h \mathbf 1_{m \times m} + ({\omega_i}^\ast) I_m$, where $M>0$ has the structure of $A^\alpha$ (see \eqref{eq:topo}), as the Lyapunov candidate. Notice that for $t \geq T$, $M=A^\alpha$ since $\boldsymbol\omega^{\sigma}$ will be constant. Thus, at jumps, $\dot{\boldsymbol \theta}$ (hence $z$) will stay unchanged. Also, $$\{\forall t \geq T, \quad \dot{V_2} \leq -\kappa \sin{\eps} z^T z\},$$
which ends the proof.
\end{proof}
\section{Simulation results}\label{sec:sim}
This section presents simulation results for a network of four oscillators over a star and a line graph topology. The initial condition for the oscillators is set to $\boldsymbol \theta(0)=[\frac{\pi}{4},\frac{\pi}{10},\frac{\pi}{2},\frac{\pi}{5}]$. We simulate both star and line networks for two sets of exogenous frequencies  $\omega=[20,3,2,1]$ and $\bar\omega=[1,10,5,6]$. For the star graph node $1$ coincides with the hub and for the line graph node $1$ and $4$ are terminal nodes. Table \ref{t1} shows the exact value of the minimum eigenvalue for each of the cases (the incidence matrix for the star graph is denoted by $B_s$ and for the line graph with $B_\ell$), together with the bound obtained based on \eqref{eq:threebounds} and a sufficient bound for $\kappa$. The latter is calculated using the exact value of $\lambda_{\min}$ reported in the second column of the table. As shown, the lower bound of $\kappa$ with $\bar\omega$, where the hub frequency is minimum, is smaller (for both star and line graph) than $\omega$ (see Section \ref{sec:sub}).\\
\begin{center}
\begin{tabu} to 0.45\textwidth { | X[l] | X[c] | X[c] | X[c] | }
\hline
Choice & $\lambda_{\min}$ & Estimation in \eqref{eq:threebounds} & $\kappa$ \\
\hline
$B_s^T \boldsymbol \omega B_s$ & 1.42  & 1 & 13.4\\
\hline
$B_s^T \boldsymbol {\bar\omega} B_s$ & 5.36  & 4 & 1.68\\
\hline
$B_\ell^T \boldsymbol \omega B_\ell$ & 1.64  & 0.58 & 10.36\\
\hline
$B_\ell^T \boldsymbol {\bar\omega} B_\ell$ & 1.64  & 0.58 & 5.48\\
\hline
\end{tabu}\\[1mm]
\small{Table 1}\label{t1}
\end{center}
Figure \ref{sim1} (Fig. \ref{sim2}) shows the relative phases and nodal frequencies for two sets of exogenous frequencies $\omega$ and ${\bar\omega}$ for a star (line) graph. We used $\kappa=5$ for all four cases. As shown in Figure \ref{sim2}, the line graph also achieves frequency synchronization with a $\kappa$ smaller than the sufficient bound obtained in Table \ref{t1}. In all cases, relative phases converge to a non-zero value and all nodal frequencies reach a consensus.
\begin{figure}[h]
\centering
\includegraphics[scale=0.22]{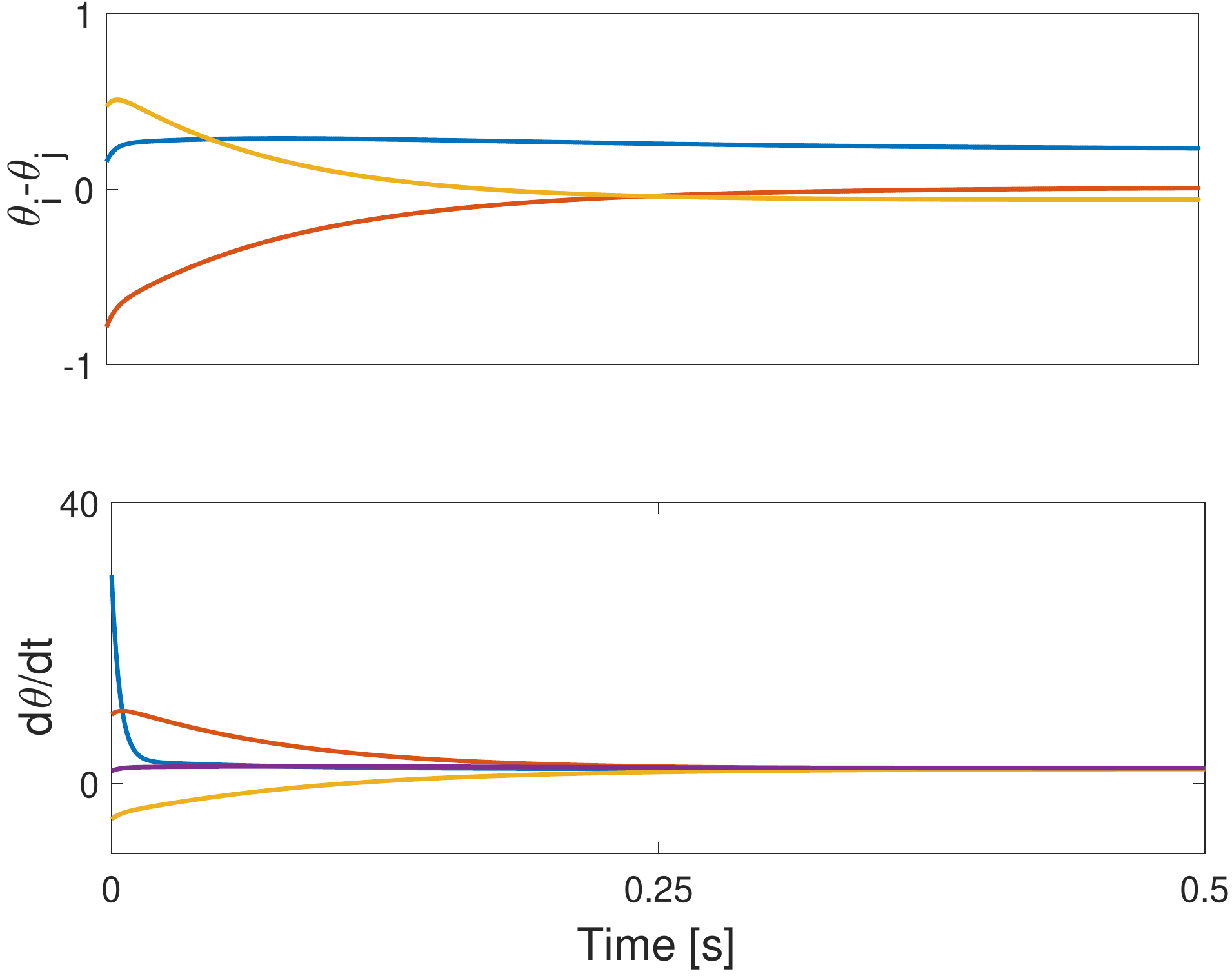}\hspace{0.5mm}\includegraphics[scale=0.22
]{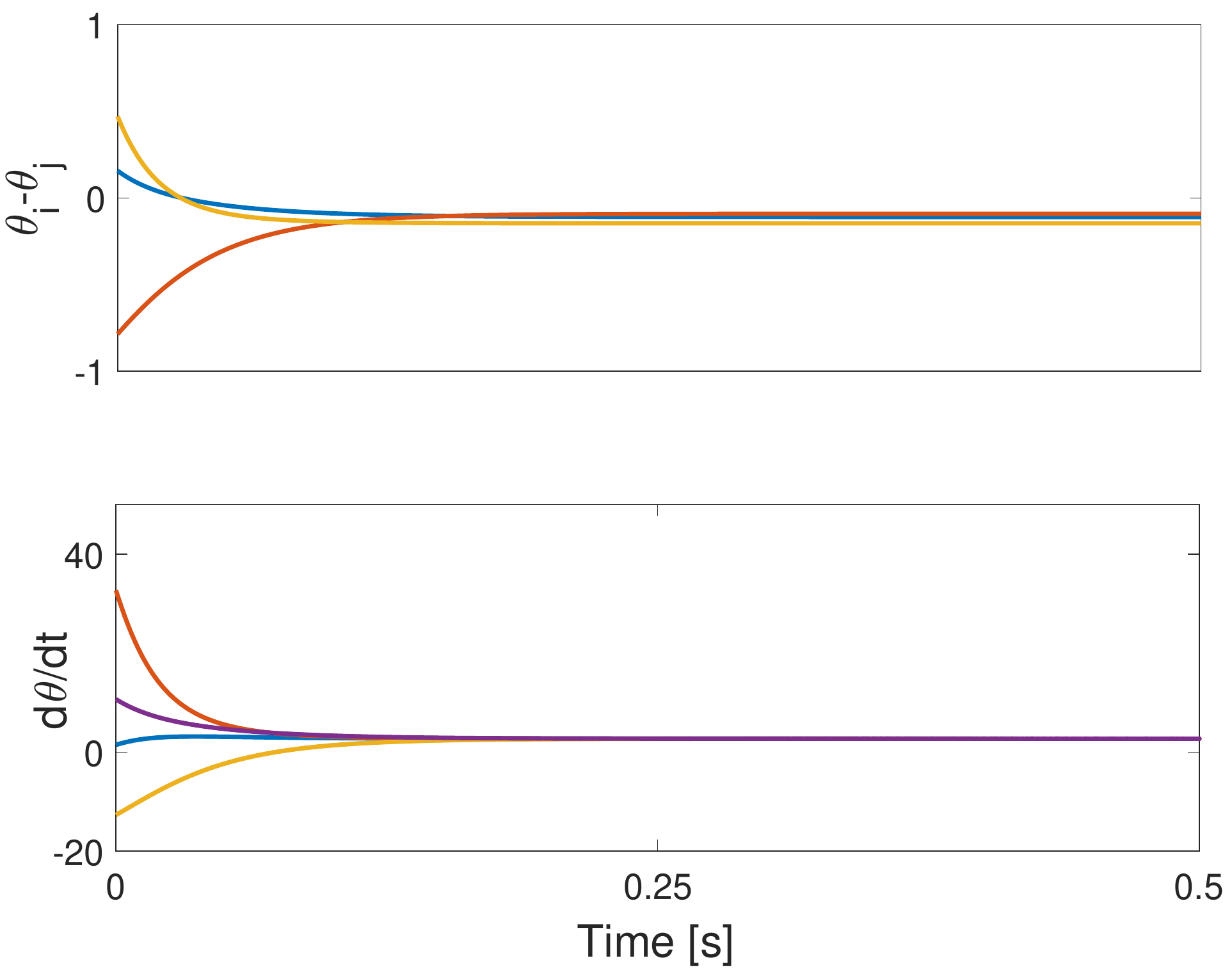}
\caption{A star graph with 4 nodes and two sets of exogenous frequencies $\boldsymbol \omega$ (left) and $\boldsymbol {\bar\omega}$ (right).}\label{sim1}
\end{figure}
\begin{figure}[h]
\centering
\includegraphics[scale=0.22]{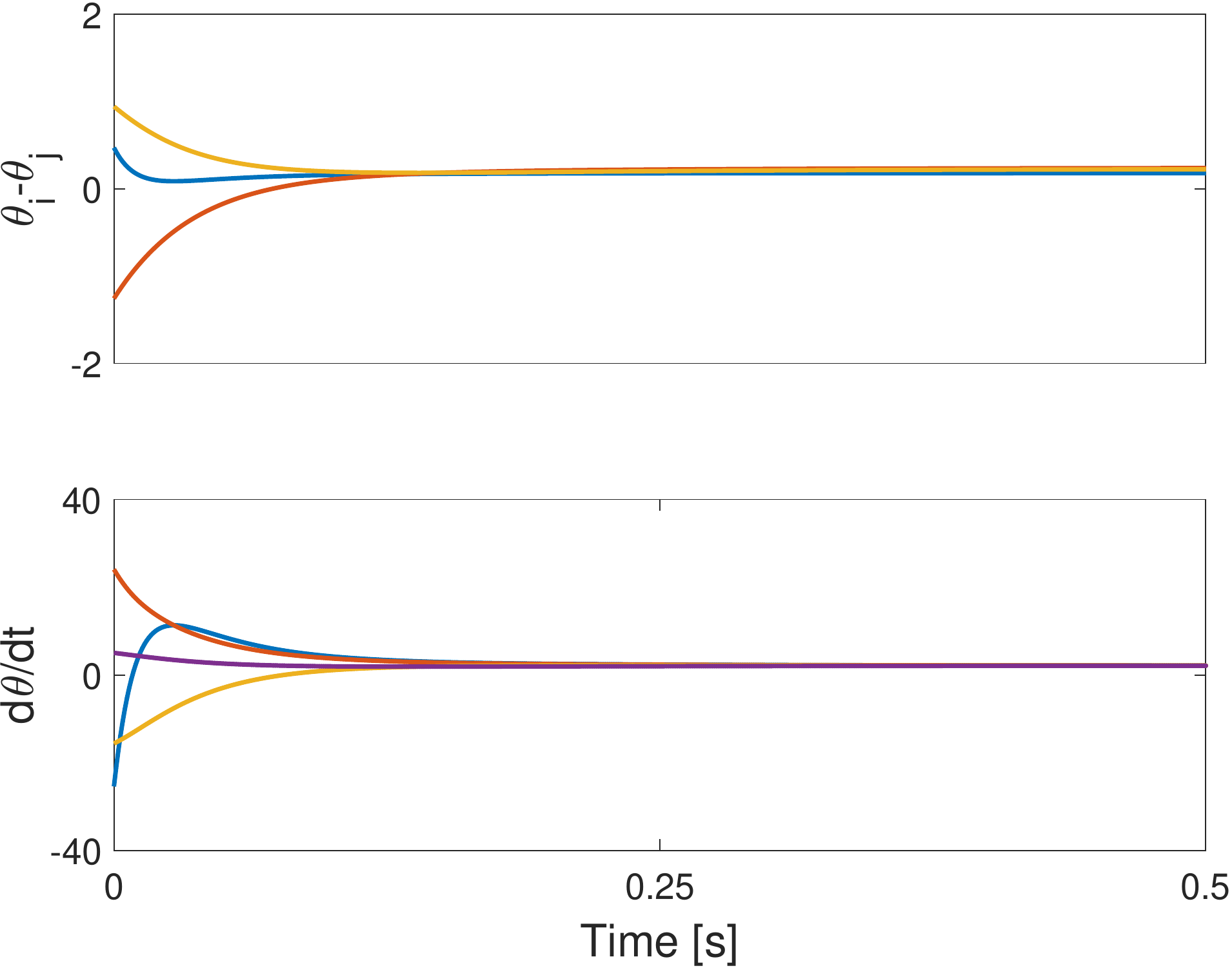}\hspace{0.5mm}\includegraphics[scale=0.22
]{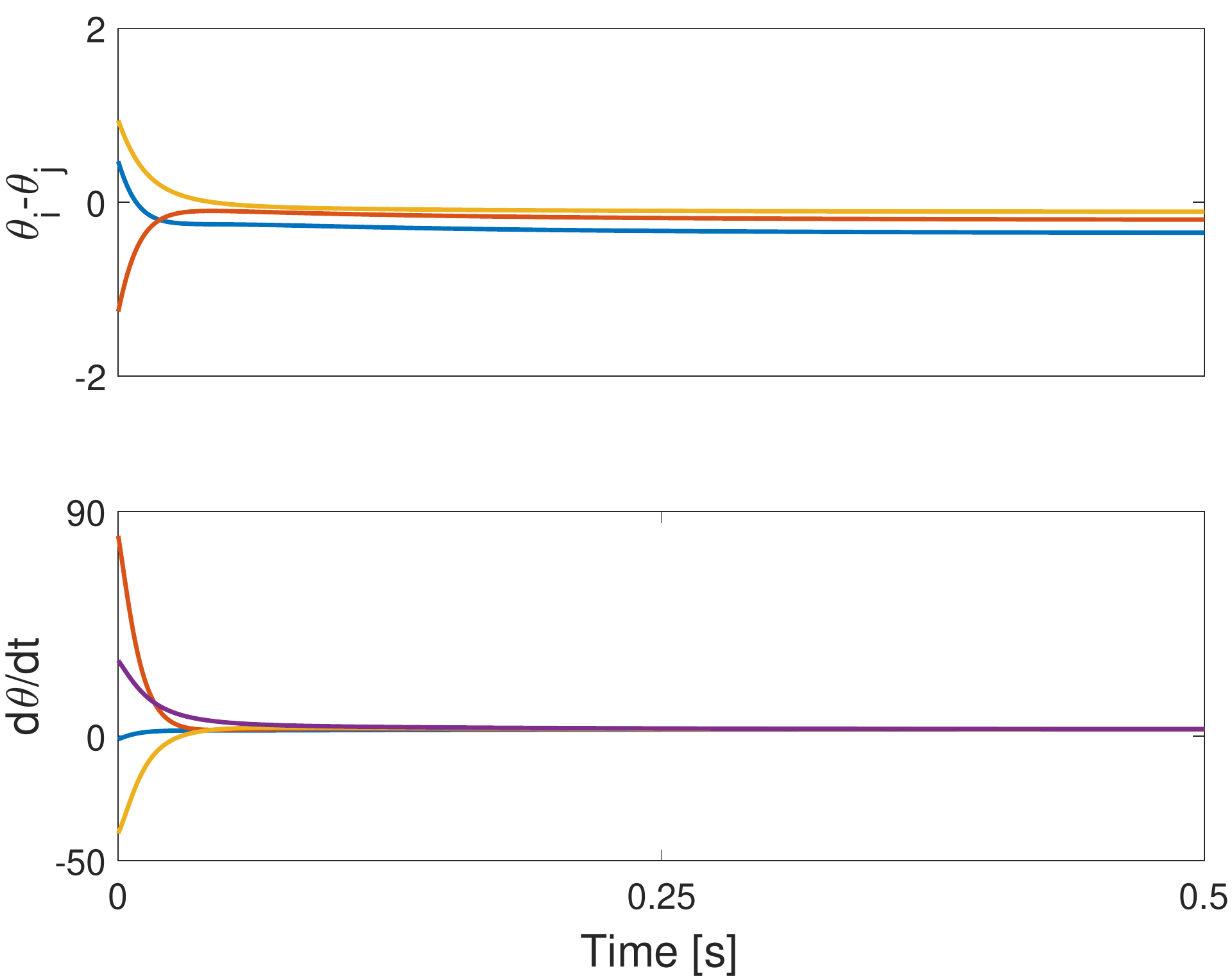}
\caption{A line graph with 4 nodes and two sets of exogenous frequencies $\boldsymbol \omega$ (left) and $\boldsymbol {\bar\omega}$ (right).}\label{sim2}
\end{figure}
To simulate the results of the event-triggered algorithm, we first take a star graph with 4 nodes with the same initial conditions as the above examples. We set $\omega=[20,18,16,6]$ and $\kappa=1.1$. The results is shown in Figure \ref{sim3}-(1). As shown the frequencies de-synchronize and the relative phases are unbounded. For the same network, we use the event-triggered control with Algorithm 1, and set $\Delta=1.1$, $\eta=\frac{\pi}{2}-\frac{\pi}{10}$, and $\eps=\frac{\pi}{10}$. The results are shown in Figure \ref{sim3}-(2). As shown the oscillators synchronize.
\begin{figure}[h]
\centering
\includegraphics[scale=0.2]{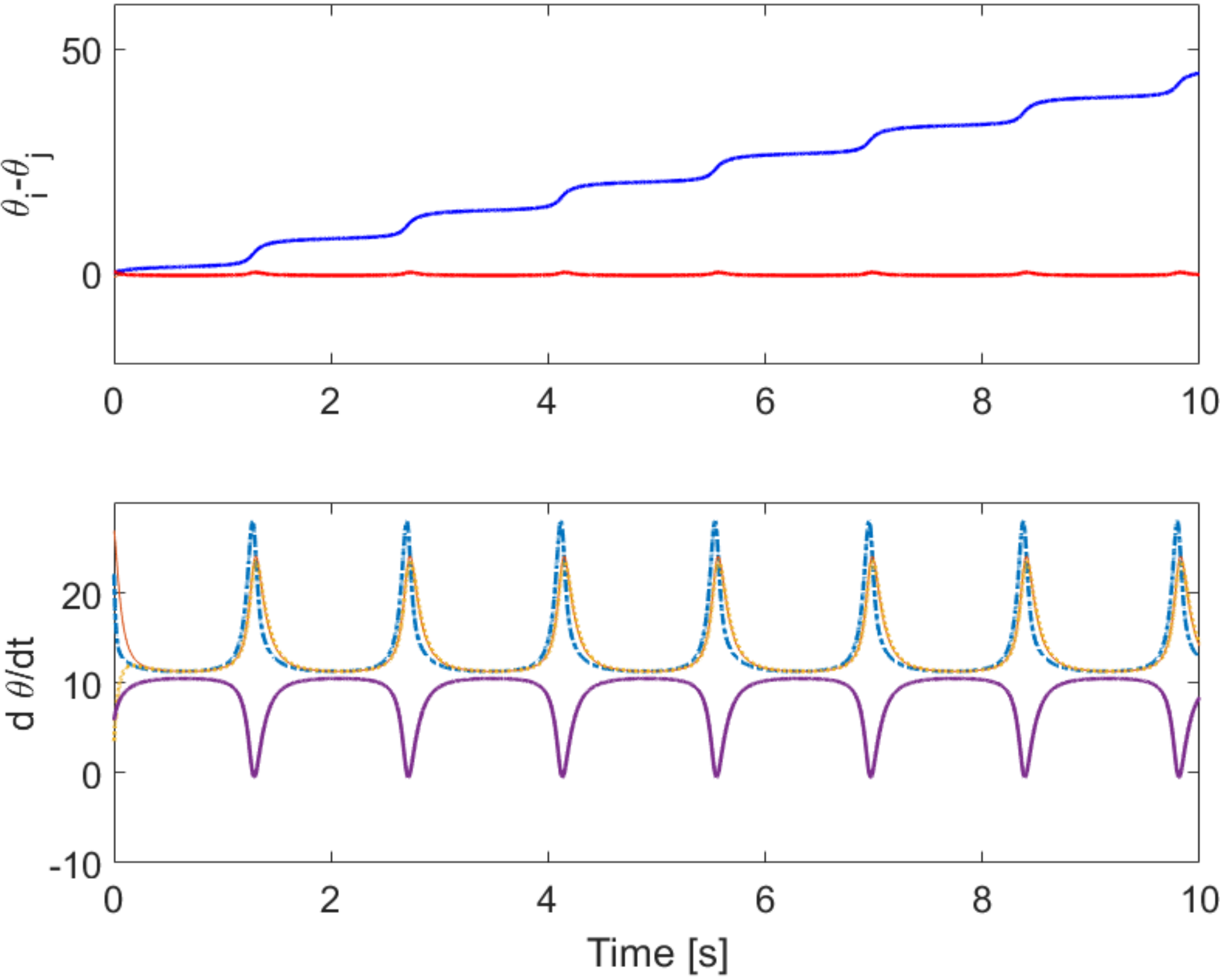}\hspace{0.1mm}\includegraphics[scale=0.21]{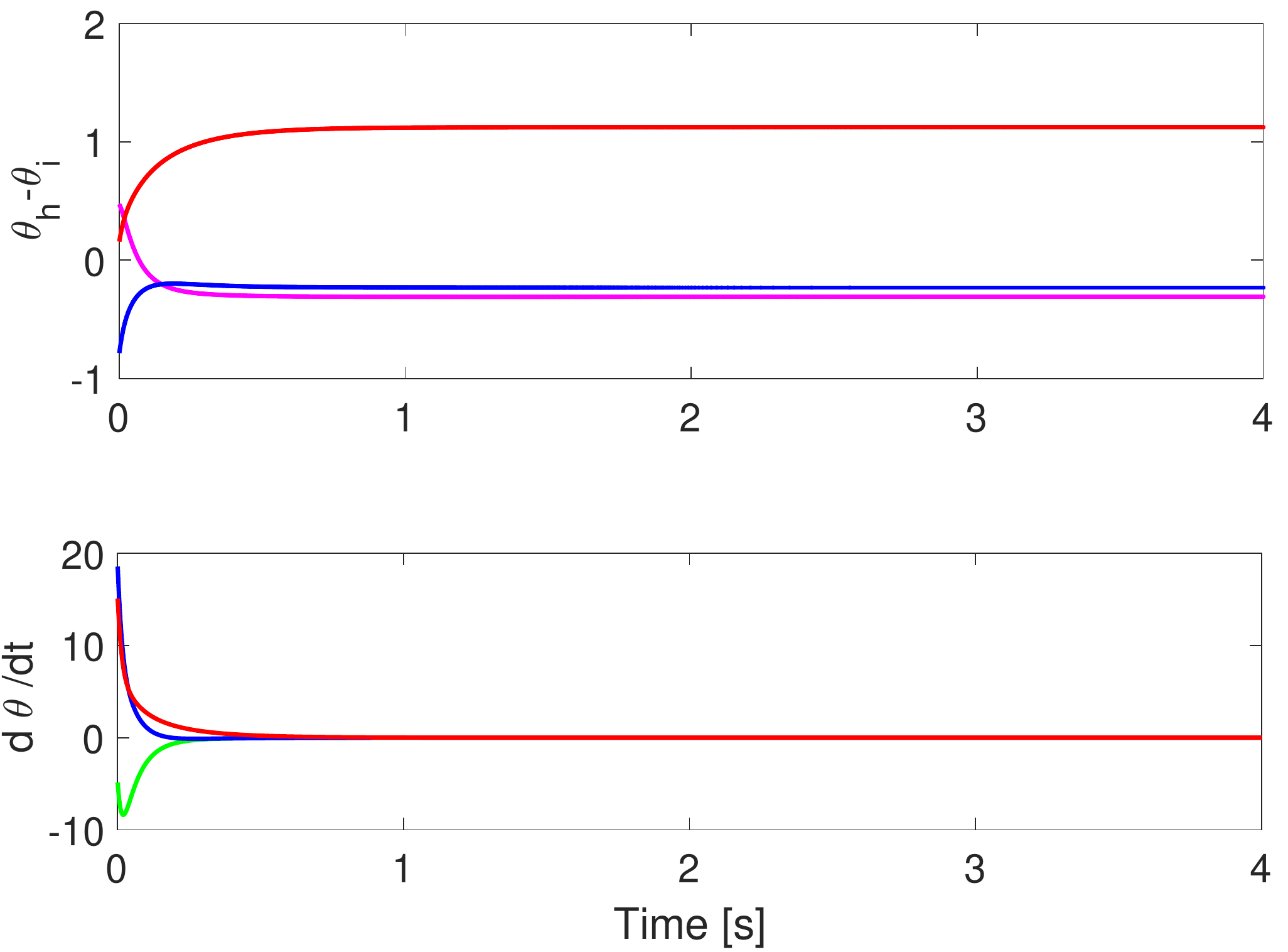}
\caption{A star graph with (1) frequency de-synchronization (left) and (2) event-triggered controller (right).}\label{sim3}
\end{figure}\\[5mm]
\section{Conclusions}\label{sec:con}
This paper has studied the synchronization of a finite number of Kuramoto oscillators in a tree network where the coupling strength of each link between every two oscillators in each direction is weighted by a common coefficient and the exogenous frequency of its corresponding head oscillator. We have driven a sufficient condition for the common coupling strength and showed its dependency on both exogenous frequencies and graph structure. We have also provided an example of the application of the obtained sufficient bound to achieve frequency synchronization in a star network for which the value of $\kappa$ was given. Future avenues include allowing zero and time-varying exogenous frequencies, and extending the event-triggered algorithm for a more general class of graphs.  
\bibliographystyle{plain}
\bibliography{biblio}
\end{document}